\documentclass[reqno]{amsart}

\makeatletter
\@namedef{subjclassname@2020}{%
  \textup{2020} Mathematics Subject Classification}
\makeatother

\usepackage{amsmath}
\usepackage{amssymb}
\usepackage{enumerate}

\newtheorem{Theorem}{Theorem}[section]
\newtheorem{Lemma}[Theorem]{Lemma}

\newtheorem{Proposition}[Theorem]{Proposition}

\newcommand{\thref}[1]{Theorem \ref{#1}}
\newcommand{\leref}[1]{Lemma \ref{#1}}
\newcommand{\reref}[1]{Remark \ref{#1}}
\newcommand{\seref}[1]{Section \ref{#1}}
\newcommand{\prref}[1]{Proposition \ref{#1}}

\theoremstyle{definition}
\newtheorem{Definition}[Theorem]{Definition}
\newtheorem{Remark}[Theorem]{Remark}

\newtheorem*{Remark*}{Remark}

\numberwithin{equation}{section}

\begin{document}

\newcommand{\pd}{\partial}
\newcommand{\res}{\mathrm{res}}
\newcommand{\alg}{\mathrm{alg}}
\newcommand{\Gr}{\mathrm{Gr}}
\newcommand{\Grad}{\mathrm{Gr^{ad}}}
\newcommand{\Ai}{\mathrm{Ai}}
\newcommand{\Span}{\mathrm{span}}
\newcommand{\clspan}{\overline{\mathrm{span}}}
\newcommand{\Id}{\mathrm{Id}}
\newcommand{\diag}{\mathrm{diag}}
\newcommand{\Ad}{\mathrm{Ad}}

\newcommand{\Pset}{\mathbb{P}}
\newcommand{\Qset}{\mathbb{Q}}
\newcommand{\Rset}{\mathbb{R}}
\newcommand{\Cset}{\mathbb{C}}
\newcommand{\Nset}{\mathbb{N}}
\newcommand{\Zset}{\mathbb{Z}}

\newcommand{\kA}{\mathfrak{A}}
\newcommand{\fK}{\mathcal{K}}
\newcommand{\fR}{\mathfrak{R}}
\newcommand{\fa}{\mathfrak{a}}
\newcommand{\fb}{\mathfrak{b}}
\newcommand{\fp}{\mathfrak{p}}

\newcommand{\cC}{\mathcal{C}}
\newcommand{\cF}{\mathcal{F}}
\newcommand{\cH}{\mathcal{H}}
\newcommand{\cK}{\mathcal{K}}
\newcommand{\cL}{\mathcal{L}}
\newcommand{\cM}{\mathcal{M}}
\newcommand{\cP}{\mathcal{P}}
\newcommand{\cQ}{\mathcal{Q}}
\newcommand{\cR}{\mathcal{R}}
\newcommand{\cS}{\mathcal{S}}
\newcommand{\cU}{\mathcal{U}}
\newcommand{\cV}{\mathcal{V}}

\newcommand{\sH}{\mathsf{H}}

\newcommand{\al}{\alpha}
\newcommand{\be}{\beta}
\newcommand{\ga}{\gamma}
\newcommand{\ka}{\kappa}
\newcommand{\la}{\lambda}
\newcommand{\om}{\omega}

\newcommand{\Pt}{\tilde{P}}
\newcommand{\tS}{\tilde{S}}
\newcommand{\Vt}{\tilde{V}}
\newcommand{\pt}{\tilde{p}}
\newcommand{\xt}{\tilde{x}}
\newcommand{\et}{\tilde{e}}
\newcommand{\mt}{\tilde{m}}
\newcommand{\kt}{\tilde{k}}
\newcommand{\nt}{\tilde{n}}
\newcommand{\wt}{\tilde{w}}
\newcommand{\Gt}{\tilde{G}}
\newcommand{\Ht}{\tilde{H}}
\newcommand{\phit}{\tilde{\phi}}
\newcommand{\thetat}{\tilde{\theta}}

\newcommand{\Lt}{\tilde{L}}

\title{Gaudin model for the multinomial distribution}

\author[P.~Iliev]{Plamen~Iliev}
\address{School of Mathematics, Georgia Institute of Technology, 
Atlanta, GA 30332--0160, USA}
\email{iliev@math.gatech.edu}
\thanks{The author gratefully acknowledges support from the Simons Foundation through grant \#635462.} 

\date{June 9, 2023}



\begin{abstract} 
The goal of the paper is to analyze a Gaudin model for a polynomial representation of the Kohno-Drinfeld Lie algebra associated with the multinomial distribution. The main result is the construction of an explicit basis of the space of polynomials consisting of common eigenfunctions of Gaudin operators in terms of Aomoto-Gelfand hypergeometric series. The construction shows that the polynomials in this basis are also common eigenfunctions of the operators for a dual Gaudin model acting on the degree indices, and therefore they provide a solution to a multivariate discrete bispectral problem.
\end{abstract}

\maketitle

\section{Introduction} \label{se1}
It is hard to overstate the role that the classical orthogonal polynomials have played in mathematics and physics over the last few centuries. Recently, there have been interesting multivariate extensions of this theory based on connections with quantum integrable systems, representation theory and algebraic combinatorics, see for instance \cite{BLV,CKY,CFR,DIVV,DX,Mac,vD} and the references therein. In particular, the work \cite{KMP2} describing the irreducible representations of the symmetry algebra of the generic superintegrable system on the $3$-sphere suggested that spectral properties of multivariate extensions of the classical orthogonal polynomials were intimately related to first integrals of superintegrable systems. The extension of this work to arbitrary dimension \cite{I3,I4} linked these results to representations of the Kohno-Drinfeld Lie algebra associated with the Dirichlet distribution. These constructions were generalized in \cite{IX3} where it was shown that one can associate discrete quantum superintegrable systems to several classical distributions whose symmetries define polynomial representations of the Kohno-Drinfeld Lie algebra. An interesting corollary of these results is that the common eigenfunctions of maximal abelian subalgebras of the Kohno-Drinfeld Lie algebra define families of multivariate orthogonal polynomials which can be regarded as natural analogs of the classical orthogonal polynomials. In the present paper we construct an explicit basis of common eigenfunctions of the algebra generated by Gaudin operators for the representation associated with the multinomial distribution. We describe these notions and outline the main results below.

The Kohno-Drinfeld Lie algebra $\mathfrak{k}_{d+1}$, which appeared in \cite{Dr,Ko}, is the quotient of the free Lie algebra on generators $L_{i,j}=L_{j,i}$, $i\neq j\in\{0,1,\dots,d\}$ by the ideal generated by the relations
\begin{subequations}\label{1.1}
\begin{align}
[L_{i,j},L_{k,l}]&=0, &&\text{ if }i,j,k,l \text{ are distinct,} \label{1.1a}\\
[L_{i,j},L_{i,k}+L_{j,k}]&=0, &&\text{ if }i,j,k \text{ are distinct.}   \label{1.1b}
\end{align}
\end{subequations}
The element
\begin{equation}\label{1.2}
\cH=\sum_{0\leq i< j\leq d} L_{i,j}
\end{equation}
is central and belongs to all maximal abelian subalgebras of $\mathfrak{k}_{d+1}$. If we think of $\cH$ as a quantum Hamiltonian, then the operators $L_{i,j}$ represent symmetries, or integrals of motion for $\cH$. For representations of the  Kohno-Drinfeld Lie algebra associated with the hypergeometric distribution and the multinomial distribution, the operator $\cH$ can be regarded as a discrete quantum superintegrable system which extends the generic quantum superintegrable system on the $d$-sphere and the quantum harmonic oscillator, respectively \cite{IX3}. Within the theory of multivariate orthogonal polynomials, the operator $\cH$ appears naturally in the characterization of the second-order partial difference operators that have discrete orthogonal polynomials as eigenfunctions \cite{IX1}. 

Two maximal abelian subalgebras of the  Kohno-Drinfeld Lie algebra have been extensively studied in the literature:
\begin{enumerate}[(i)]
\item \label{ii1}
The abelian algebras generated by the {\em Jucys-Murphy elements}
$$L_{0,1}, L_{0,2}+L_{1,2}, L_{0,3}+L_{1,3}+L_{2,3},\dots, \sum_{j=0}^{d-1}L_{j,d},$$
which play an important role in the representation theory of the symmetric group \cite{OV}.
\item \label{ii2}
The abelian algebras generated by the {\em Gaudin elements}
\begin{equation}\label{1.3}
G_i(\al)=\sum_{\begin{subarray}{c}j=0\\ j\neq i \end{subarray}}^{d}\frac{L_{i,j}}{\al_{i}-\al_{j}},
\end{equation}
where  $\al_0,\dots,\al_{d}$ are fixed distinct numbers, which were considered in the work of Gaudin \cite{Gau1} for specific representations of  $\mathfrak{k}_{d+1}$. 
\end{enumerate}

For representations of the Kohno-Drinfeld Lie algebra associated with the multinomial, Dirichlet and Hahn distributions, the common eigenfunctions of the Jucys-Murphy elements in (\ref{ii1}) can be written as products of hypergeometric functions. Explicit formulas for these families of orthogonal polynomials and their bispectral properties can be found in \cite{GI}. The main result in this work is the construction of a basis of common eigenfunctions of Gaudin operators  (\ref{ii2}) in terms of Aomoto-Gelfand hypergeometric series for the representation of $\mathfrak{k}_{d+1}$ associated with the multinomial distribution. The Bethe ansatz equations in \cite{Gau2} are replaced here by a simpler decoupled system of equations for the parameters of a dual Gaudin model which is diagonalized by the same polynomials considered as functions of their degree indices. This establishes the bispectrality of the multinomial Gaudin model in the sense of Duistermaat and Gr\"unbaum \cite{DG}. Throughout the paper we focus on the multinomial distribution, but the constructions can be easily extended to the negative multinomial distribution, see \reref{re2.2} for details. 

The paper is organized as follows. In the next section we define the representation of the Kohno-Drinfeld Lie algebra for the multinomial distribution. We also review some of the results in \cite{I1}, and in particular, the set $\fK_{d}$ whose points parametrize Krawtchouk polynomials in $d$ variables and their bispectral properties which play a crucial in the proof of the main result. In \seref{se3} we describe several algebraic properties of the  operators $L_{i,j}$ associated with the multinomial distribution. We also derive necessary and sufficient conditions for the multivariate Krawtchouk polynomials to be common eigenfunctions of the Gaudin operators. 
This leads to a complicated overdetermined system of nonlinear algebraic equations for the free parameters defining the point $\ka\in \fK_{d}$. In \seref{se4}, we state and prove the main result of the paper by constructing a solution of the nonlinear equations using an appropriate ansatz.  

\section{Multivariate Krawtchouk polynomials}\label{se2}

\subsection{Representations of the  Kohno-Drinfeld algebra for the multinomial distribution}
Suppose that $p_0,p_1,\dots,p_d$ are positive real numbers such that 
\begin{equation}\label{2.1}
p_0+p_1+\cdots+p_d=1.
\end{equation}
The probability mass function of the multinomial distribution with parameters $p=(p_0,p_1,\dots,p_d)$ and $N\in\Nset$ is
$$\binom{N}{x_0,x_1,\dots,x_d} p_0^{x_0}p_1^{x_1}\cdots p_{d}^{x_{d}} =\frac{N!}{x_0!x_1!\cdots x_{d}!}\, p_0^{x_0}p_1^{x_1}\cdots p_{d}^{x_{d}},$$
where $x_i\in\Nset_0 $ and $x_0+\cdots+x_{d}=N$.
We set $x_{0}=N-(x_1+\cdots+x_{d})$, and throughout the paper we work with the independent variables $x=(x_1,\dots,x_d)$. This leads to the weight
$$W_{p,N}(x)=\binom{N}{N-|x|,x_1,\dots,x_d} \, p_{0}^{N-|x|} p_1^{x_1}\cdots p_{d}^{x_{d}} , $$
where $x=(x_1,\dots,x_d)\in V_N^d= \{x \in \Nset_0^d: |x|=x_1+\cdots+x_{d} \le N\}$ and the corresponding inner product
\begin{equation}\label{2.2}
\langle f,g\rangle_{p,N}=\sum_{x\in V_N^d}f(x)g(x)W_{p,N}(x).
\end{equation}
Let $\{e_1,e_2,\dots,e_d\}$ be the standard basis for $\Rset^d$, and let $E_{x_i}$ and $E_{x_i}^{-1}$ denote the shift operators
\begin{align*}
E_{x_i}f(x)=f(x+e_i) \quad \text{ and }\quad E_{x_i}^{-1}f(x)=f(x-e_i).
\end{align*}
Following \cite{IX3}, for $i\neq j\in\{0,\dots,d\}$ we define
\begin{equation}\label{2.3}
L_{i,j} =  p_i x_j (E_{x_i}E_{x_j}^{-1}-\Id)+ p_j x_i (E_{x_j}E_{x_i}^{-1}-\Id),
\end{equation}
with the convention that $E_{x_{0}}=\Id$ and $x_{0}=N-(x_1+\cdots+x_{d})$. One can show that these operators satisfy equations \eqref{1.1}, i.e.  they define a representation of the Kohno-Drinfeld Lie algebra. Moreover,  the operators $L_{i,j}$ have the following properties:
\begin{enumerate}
\item They are self-adjoint with respect to the inner product \eqref{2.2}.
\item For every $k\in\Nset_0$, they preserve the space 
$$\Rset_k[x]=\Span\{x_1^{m_1}\cdots x_d^{m_d}:m_1+\cdots+m_d\leq k\}$$ 
of polynomials of total degree at most $k$, i.e. $L_{i,j}:\Rset_{k}[x]\to\Rset_{k}[x]$.
\end{enumerate}
Thus, for every $k\in\Nset_0$, $k\leq N$, the operators in \eqref{2.3} define a representation of $\mathfrak{k}_{d+1}$ on the space 
$$\cP_k(p;N)=\Rset_{k}[x]\ominus\Rset_{k-1}[x]$$ 
of polynomials of degree $k$ which are orthogonal to all polynomials of degree at most $k-1$ with respect to the inner product \eqref{2.2} induced by the multinomial distribution.

\subsection{Parametrization and spectral properties of multivariate Krawtchouk polynomials}
Orthogonal bases of the space $\cP_k(p;N)$ were introduced by Griffiths \cite{Gr1} using generating functions. Mizukawa and Tanaka \cite{MT} gave an explicit formula for these polynomials in terms of the Aomoto-Gelfand hypergeometric series \cite{AK,Gel}. An interpretation of the polynomials within the context of the Lie algebra $\mathfrak{sl}_{d+1}$ was given in \cite{I1}. This approach provided yet another proof of the orthogonality and established their bispectral properties by showing that the polynomials are eigenfunctions of commuting partial difference operators parametrized by Cartan subalgebras of  $\mathfrak{sl}_{d+1}$. Connections to the Racah algebra, numerous probabilistic applications, Lax pairs and the analysis of the reproducing kernel of the multinomial distribution can be found in the recent works \cite{CVV,DiGr,GK,Xu}.
We review below the construction of these polynomials together with the commutative algebras of partial difference operators diagonalized by them which will be needed later, following the approach in \cite{I1}.

\begin{Definition}\label{de2.1} 
Let  $\fK_{d}$ denote the set of 4-tuples $(\nu,P,\Pt,U)$, where $\nu$ is a nonzero real number 
and $P,\Pt, U$ are $(d+1)\times (d+1)$ matrices with real entries satisfying the following conditions:
\begin{itemize}
\item[{(i)}] $P=\diag(p_0,p_1,\dots,p_{d})$ and 
$\Pt=\diag(\pt_0,\pt_1,\dots,\pt_{d})$ are diagonal, and
$p_0=\pt_0=\frac{1}{\nu}$;
\item[{(ii)}] $U=(u_{i,j})_{0\leq i,j\leq d}$ is such that 
$u_{0,j}=u_{j,0}=1$ for all $j=0,1,\dots,d$, i.e. 
\begin{equation}\label{2.4}
U=\left(\begin{matrix}1 & 1 & 1 & \dots &1\\ 
1 & u_{1,1} &u_{1,2} &\dots & u_{1,d}\\
 \vdots & \\
 1 & u_{d,1} & u_{d,2}  &\dots &u_{d,d}
\end{matrix}\right);
\end{equation}
\item[{(iii)}]  The following matrix equation holds
\begin{equation}\label{2.5}
\nu PU\Pt U^{t}=I_{d+1}. 
\end{equation}
\end{itemize}
\end{Definition}

From the definition it is easy to see that $p_j$ and $\pt_j$ are nonzero numbers, such that 
$$\sum_{j=0}^{d}p_j=\sum_{j=0}^{d}\pt_j=1.$$

For every point $\kappa\in \fK_{d}$ and every positive integer $N$ we define $d$-variable Krawtchouk polynomials $P_{n}(x;\kappa,N)$ with indices $n=(n_{1},\dots,n_{d})\in V_N^d$ in the variables $x=(x_{1},\dots,x_{d})\in  V_N^d$ in terms of the Aomoto-Gelfand hypergeometric series of type $(d+1,2d+2)$
\begin{equation}\label{2.6}
P_{n}(x;\kappa,N)=\sum_{A=(a_{i,j})\in \mathcal{M}_{d,N}}\frac{\prod_{j=1}^{d}(-n_j)_{\sum_{i=1}^{d}a_{i,j}}\, \prod_{i=1}^{d}(-x_i)_{\sum_{j=1}^{d}a_{i,j}}}{(-N)_{\sum_{i,j=1}^{d}a_{i,j}}}
\prod_{i,j=1}^{d}\frac{(1-u_{i,j})^{a_{i,j}}}{a_{i,j}!}.
\end{equation}
In the above formula $(b)_j$ is the Pochhammer symbol
$$(b)_0=1,\qquad \qquad (b)_j = b(b+1)\ldots(b+j-1)\quad\text{ for }j\in\Nset,$$ 
and $\mathcal{M}_{d,N}$ denotes the set of all $d\times d$ matrices $A=(a_{i,j})$ with nonnegative integer entries such that $\sum_{i,j=1}^{d}a_{i,j}\leq N$. 

The polynomials in \eqref{2.6} are mutually orthogonal with respect to the multinomial distribution:
\begin{equation}\label{2.7}
\langle P_{n}(x;\kappa,N) ,P_{m}(x;\kappa,N)\rangle_{p,N}=\frac{p_0^N}{W_{\pt,N}(n)}\, \delta_{n,m},
\end{equation}
and therefore $\{P_{n}(x;\kappa,N):|n|=k\}$ is an orthogonal basis of $\cP_k(p;N)$.

If we start with fixed positive real numbers $p_0,p_1,\dots,p_d$, we can construct a point $\kappa\in\fK_{d}$ as follows. First, we construct a matrix $U$ of the form in \eqref{2.4} such that $U^tPU$ is diagonal, i.e. 
\begin{equation}\label{2.8}
\sum_{j=0}^{d}p_ju_{j,i}u_{j,k}=0 \qquad \text{ for all }i\neq k\in\{0,1,\dots,d\}.
\end{equation}
In other words, the columns $w_j=(1,u_{1,j},\dots,u_{d,j})^t$ of $U$ must be mutually orthogonal vectors in $\Rset^{d+1}$ with respect to the inner product $(v,w)=v^tPw$ and $w_0=(1,1,\dots,1)^t$. 
The diagonal entries of $\Pt$ are uniquely determined from $p$ and $U$ by
\begin{equation}\label{2.9}
\pt_k=\frac{p_0}{\sum_{j=0}^{d}p_ju_{j,k}^2}\qquad  \text{ for } k=0,\dots,d.
\end{equation}
In dimension one, the matrix $U$ and the point $\kappa\in\fK_{1}$ are uniquely determined from $p_1$, and \eqref{2.6} gives a formula for the polynomials   introduced by Krawtchouk \cite{Kra}. However, when $d>1$, there are $d(d-1)/2$ degrees of freedom in choosing $U$ from the parameters $p_0,\dots,p_d$. 

The polynomials $P_{n}(x;\kappa,N)$ provide solutions to a multivariate discrete-discrete bispectral problem \cite{DG}. More precisely, using the operators \eqref{2.3} and the analogous operators 
\begin{equation}\label{2.10}
\Lt_{i,j} = \pt_i n_j (E_{n_i}E_{n_j}^{-1}-\Id)+ \pt_j n_i (E_{n_j}E_{n_i}^{-1}-\Id),
\end{equation}
acting on the indices $n=(n_{1},\dots,n_{d})$, with the convention that  $E_{n_{0}}=\Id$ and $n_{0}=N-|n|$, we have
\begin{subequations}\label{2.11}
\begin{align}
&n_{i}P_{n}(x;\kappa,N)= \frac{\pt_{i}}{p_{0}} \left[ \sum_{0\leq k<l\leq d} u_{k,i}u_{l,i}L_{k,l} \right] P_{n}(x;\kappa,N), \label{2.11a}\\
&x_{i}P_{n}(x;\kappa,N)= \frac{p_{i}}{p_{0}} \left[ \sum_{0\leq k<l\leq d} u_{i,k}u_{i,l}\Lt_{k,l} \right] P_{n}(x;\kappa,N), \label{2.11b}
\end{align}
\end{subequations}
see \cite[Theorem 6.1]{I1}.
The bispectral equations above can be naturally related by a bispectral involution. Indeed, note that the mapping 
\begin{subequations}\label{2.12}
\begin{equation}\label{2.12a}
\kappa=(\nu,P,\Pt,U)\to\tilde{\kappa}=(\nu,\Pt,P,U^{t})
\end{equation}
defines an involution on $\fK_{d}$. Moreover, from \eqref{2.6} it follows that 
\begin{equation}\label{2.12b}
P_{n}(x;\kappa,N)=P_{x}(n;\tilde{\kappa},N),
\end{equation}
\end{subequations}
i.e. the involution exchanges the roles of the variables $x$ and the indices $n$, thus relating equations \eqref{2.11} above.\\

\begin{Remark}\label{re2.2}
We focus in the paper on the multinomial distribution and therefore we will assume that the parameters $p_0,\dots,p_d$ are positive real numbers satisfying \eqref{2.1}, $N$ is a positive integer and we will work with polynomials with real coefficients. However, most of the constructions can be easily extended for generic real or complex numbers $p_j$ and  $N$ satisfying \eqref{2.1}. In that case, we define the polynomials by formula \eqref{2.6}, where the sum on the right-hand side is over all matrices with entries in $\Nset_0$ (note that the series is terminating because $n_j\in\Nset_0$). Thus, we obtain polynomials defined for all $n\in\Nset_0^d$ and equations~\eqref{2.11} hold. In particular, if we pick positive real numbers $s,c_1,\dots,c_d$ such that $|c|=c_1+\cdots+c_d<1$ and if we set formally $N= -s$ and $p_j = - c_j /(1-|c|)$ for $j=1,\dots,d$, we obtain polynomials defined for all $n\in\Nset_0^d$ which are mutually orthogonal with respect to the negative multinomial distribution
$$W(x;c,s)=(1-|c|)^{s}(s)_{|x|}\prod_{j=1}^{d}\frac{c_j^{x_j}}{x_j!}.$$ 
These polynomials were introduced by Griffiths \cite{Gr2} in terms of a generating function, and a direct combinatorial proof of their hypergeometric representation and bispectral properties were obtained in \cite{I2}.
\end{Remark}

\subsection*{Notations} 
Throughout the paper, it will be convenient to work with a linear combination of the Gaudin elements in \eqref{1.3}. We set
\begin{equation}\label{2.13}
G(\al,p,N;\zeta)=\sum_{i=0}^{d}\zeta_{i} G_i(\al)=\sum_{0\leq i<j\leq d}\frac{\zeta_{i}-\zeta_{j}}{\al_{i}-\al_{j}}L_{i,j},
\end{equation}
where $L_{i,j}$ are the operators acting on the variables $x_{i}$ given in \eqref{2.3}, and we denote by $\mathfrak{G}_{d+1}(\al,p,N)$ the abelian algebra generated by the operators in \eqref{2.13}. 

Similarly, we denote by
\begin{equation}\label{2.14}
\Gt(\al,\pt,N;\zeta) =\sum_{0\leq i<j\leq d}\frac{\zeta_{i}-\zeta_{j}}{\al_{i}-\al_{j}}\Lt_{i,j},
\end{equation}
the operators acting on the degree indices $n_{i}$, where $\Lt_{i,j}$ are defined in \eqref{2.10}, and by $\tilde{\mathfrak{G}}_{d+1}(\al,\pt,N)$ the  abelian algebra generated by the operators in \eqref{2.14}.

\section{Algebraic properties of $L_{i,j}$ and conditions for diagonalization}\label{se3}

\begin{Lemma}\label{le3.1}
The operators $\{ L_{i,j} \}_{{0\leq i< j\leq d}}$ are linearly independent as operators acting on the space $\Rset_{1}[x]=\Span\{1,x_{1},\dots,x_{d}\}$. 
\end{Lemma}

\begin{proof}
Let
\begin{equation}\label{3.1}
M=\sum_{0\leq i < j\leq d} c_{i,j}L_{i,j} \qquad \text{ where }\qquad c_{i,j}\in\Rset
\end{equation}
be a linear combination of the operators $L_{i,j}$. To simplify the notation, we define $c_{i,j}=c_{j,i}$ when $i>j$. Suppose that $M=0$. Then 
\begin{equation*}
M(x_{i})=0
\end{equation*}
for every $i\in \{1,\dots,d\}$. Note that $L_{k,j}(x_{i})=0$ when $i\notin \{k,l\}$ and  $L_{i,j}(x_{i})=p_{i}x_{j}+p_{j}x_{i}$, where $x_{0}=N-|x|$. Since the constant term of $M(x_{i})$ is $c_{0,i}Np_{i}$ we see that $c_{0,i}=0$. For $j\in \{1,\dots,d\} \setminus \{i \}$ the coefficient of $x_{j}$ in $M(x_{i})$ is $(c_{i,j}-c_{0,i})p_{i}$, hence $c_{i,j}=0$. This shows that $c_{i,j}=0$ for all $i\neq j \in\{0,1,\dots,d\}$, completing the proof.
\end{proof}

\begin{Remark}
While the operators $\{ L_{i,j} \}_{{0\leq i< j\leq d}}$ are linearly independent, they satisfy nontrivial algebraic relations. Indeed, if $i,j,k,m$ are distinct indices, then
\begin{equation*}
p_kp_m L_{i,j}=\left[L_{i,k},[L_{k,m},L_{j,m}]\right]+p_jp_k L_{i,m}+p_ip_mL_{j,k}-p_ip_jL_{k,m}.
\end{equation*}
In particular, this relation shows that the $2d-1$ operators 
\begin{equation}\label{3.2}
\{\cL_{0,i}:i=1,2,\dots,d\} \cup \{ L_{1,j}:j=2,3,\dots,d\} 
\end{equation}
are sufficient to generate the image of the Kohno-Drinfeld Lie algebra under the multinomial representation.
By taking appropriate limits, one can reduce the operator 
\begin{equation}\label{3.3}
\cH=\sum_{0\leq i< j\leq d} L_{i,j}
\end{equation}
to the Hamiltonian of the quantum harmonic oscillator, see \cite[Section 3.4]{IX3}. Thus, we can consider the system with Hamiltonian $\cH$ in \eqref{3.3} as a discrete quantum integrable extension of the harmonic oscillator, whose symmetry algebra is generated by the first integrals given in \eqref{3.2}.

Abelian subalgebras of Lie algebras with generators satisfying the Kohno-Drinfeld relations \eqref{1.1} of maximal dimension in the space $\mathfrak{k}_{d+1}^1$ spanned by $\{L_{i,j}\}_{0\leq i<j\leq d}$ were studied in a detail in \cite{AFV}. In particular, the authors show that if the elements $\{L_{i,j}\}_{0\leq i<j\leq d}$ and the brackets  $\{[L_{i,j},L_{j,k}]\}_{0\leq i<j<k\leq d}$ are linearly independent, then the maximal abelian subalgebras form a nonsingular irreducible projective subvariety of the Grassmannian of $d$-planes in $\mathfrak{k}_{d+1}^1$ isomorphic to the moduli space $\overline{\cM}_{0,d+2}$ of stable curves of genus zero with $d+2$ marked points. 
However, for $d\geq 3$  it is not hard to check that the operators in \eqref{2.10} satisfy the following relation
$$ p_{0}[L_{1,2},L_{2,3}]- p_{1}[L_{0,2},L_{2,3}]+p_{2}[L_{0,1},L_{1,3}] -p_{3}[L_{0,1},L_{1,2}]=0,$$
and therefore the brackets are not linearly independent.
\end{Remark}
The next proposition gives necessary and sufficient conditions for the multivariate Krawtchouk polynomials \eqref{2.6} to be common eigenfunctions of the Gaudin operators.
\begin{Proposition}\label{pr3.2}
For fixed distinct numbers $\al_{0},\dots,\al_{d}$ and $\kappa=(\nu,P,\Pt,U)\in\fK_{d}$ the following conditions are equivalent.
\begin{enumerate}[\rm(a)]
\item  \label{pr3.1a}  There exists $\la_{n}(\al;\zeta)\in\Rset$ such that
\begin{equation}\label{3.4}
G(\al,p,N;\zeta)P_{n}(x;\kappa,N)=\la_{n}(\al;\zeta)P_{n}(x;\kappa,N) \text{ for all }\zeta\in\Rset^{d+1} \text{ and }n\in V_N^d.
\end{equation}
\item \label{pr3.1b} The following identities hold
\begin{subequations}\label{3.5}
\begin{align}
&\frac{\al_{k}-\al_{0}}{\al_{k}-\al_{l}}=\frac{p_{k}}{p_{0}}\sum_{j=0}^{d}\pt_{j}u_{k,j}^{2}u_{l,j}, &\text{ for all }\quad k\neq l \in\{1,\dots,d\},&\label{3.5a}\\
&\sum_{j=0}^{d}\pt_{j}u_{i,j}u_{k,j}u_{l,j}=0, &\text{ for all distinct }\quad i,k,l \in \{1,\dots,d\}.&\label{3.5b}
\end{align}
\end{subequations}
\end{enumerate}
Moreover, if the equivalent conditions \rm{(\ref{pr3.1a})-(\ref{pr3.1b})} above hold then 
\begin{equation}\label{3.6}
\la_{n}(\al;\zeta) =-\sum_{i=1}^{d}\sum_{j=1}^{d} \frac{\zeta_{i}-\zeta_{0}}{\al_{i}-\al_{0}}p_{i}n_{j}(1-u_{i,j}).
\end{equation}
\end{Proposition}

\begin{proof}
We show first that if (\ref{pr3.1a}) holds, then $\la_{n}(\al;\zeta)$ is given by \eqref{3.6}. Since $P_{n}(0;\kappa,N)=1$ we can compute $\la_{n}(\al;\zeta)$ by evaluating the left-hand side of \eqref{3.4} at $x=0$. From \eqref{2.3} it is clear that if $0<i<j\leq d$ then $L_{i,j}(q(x))|_{x=0}=0$ for any polynomial $q(x)$. Thus, 
\begin{equation}\label{3.7}
G(\al,p,N;\zeta)P_{n}(x;\kappa,N)\Big|_{x=0}=\sum_{i=1}^{d}\frac{\zeta_{i}-\zeta_{0}}{\al_{i}-\al_{0}}L_{0,i}(P_{n}(x;\kappa,N)) \Big|_{x=0},
\end{equation}
and we need to evaluate $L_{0,i}(P_{n}(x;\kappa,N))$ at $x=0$. Since 
$$L_{0,i} =p_0 x_i (E_{x_i}^{-1}-\Id)+ p_i (N-|x|) (E_{x_i}-\Id)$$ 
it follows that 
\begin{equation}\label{3.8}
L_{0,i}(P_{n}(x;\kappa,N)) \Big|_{x=0}=p_{i}N [(E_{x_i}-\Id) P_{n}(x;\kappa,N)]  \Big|_{x=0}.
\end{equation}
Note that 
$$(E_{x_i}-\Id) (-x_{i})_{\ell_{i}}=-\ell_{i}(-x_{i})_{\ell_{i}-1},\quad \text{ and therefore }\quad 
(E_{x_i}-\Id) (-x_{i})_{\ell_{i}}  \Big|_{x_i=0}=-\delta_{\ell_{i},1}.$$
This shows that 
$$(E_{x_i}-\Id)[ (-x_{1})_{\ell_{1}}\cdots  (-x_{d})_{\ell_{d}}]  \Big|_{x=0}=-\delta_{\ell , e_i}
=\begin{cases}-1 &\text{ if }\ell_i =1 \text{ and }\ell_j=0 \text{ for }j\neq i,\\
0 &\text{ otherwise.}
\end{cases}
$$
Using the above and the explicit formula \eqref{2.6} for $P_{n}(x;\kappa,N)$ we see that we can have nonzero contributions in 
$[(E_{x_i}-\Id) P_{n}(x;\kappa,N)]  \big|_{x=0}$ only when $A=A_{i,j}$ for some $j\in \{1,\dots, d\}$, where $A_{i,j}$ denotes the $d\times d$ matrix that has $(i,j)$th entry $1$ and all other entries zeros. Moreover, if $A=A_{i,j}$ we have
\begin{align*}
&(E_{x_i}-\Id)\left[ \frac{\prod_{k=1}^{d}(-n_k)_{\sum_{l=1}^{d}a_{l,k}}\, \prod_{l=1}^{d}(-x_l)_{\sum_{k=1}^{d}a_{l,k}}}{(-N)_{\sum_{l,k=1}^{d}a_{l,k}}}\right] \prod_{l,k=1}^{d}\frac{(1-u_{l,k})^{a_{l,k}}}{a_{l,k}!} \\
 &\qquad\qquad=(E_{x_i}-\Id)\left[-\frac{n_jx_i}{N}\right] (1-u_{i,j})=-\frac{n_j(1-u_{i,j})}{N},
\end{align*}
which combined with \eqref{3.7}-\eqref{3.8} yields \eqref{3.6}.

We show next that (a) and (b) are equivalent. Substituting the explicit formula given in \eqref{3.6} for $\la_{n}(\al;\zeta)$ into \eqref{3.4} and replacing $n_jP_{n}(x;\kappa,N)$ by the operator on the right-hand side of \eqref{2.11a} we can rewrite the right-hand side of \eqref{3.4} as follows
\begin{equation*}
\begin{split}
&\la_{n}(\al;\zeta)P_{n}(x;\kappa,N)\\
&\qquad=-\left[ \sum_{i=1}^{d}\sum_{j=1}^{d} \frac{\zeta_{i}-\zeta_{0}}{\al_{i}-\al_{0}}p_{i}(1-u_{i,j}) 
 \frac{\pt_{j}}{p_{0}}\sum_{0\leq k<l\leq d} u_{k,j}u_{l,j}L_{k,l}\right] P_{n}(x;\kappa,N).
 \end{split}
\end{equation*}
Using the last equation, \eqref{2.13} and \leref{le3.1} we see that \eqref{3.4} holds if and only if 
\begin{equation}\label{3.9}
\begin{split}
&\frac{\zeta_{k}-\zeta_{l}}{\al_{k}-\al_{l}}=- \sum_{i=1}^{d}\sum_{j=1}^{d} \frac{\zeta_{i}-\zeta_{0}}{\al_{i}-\al_{0}}p_{i}(1-u_{i,j}) 
 \frac{\pt_{j}}{p_{0}} u_{k,j}u_{l,j},\\
&\qquad\qquad \text{ for all $\zeta=(\zeta_0,\dots,\zeta_d)\in\Rset^{d+1}$ and $k\neq l\in \{0,\dots,d\}$. }
\end{split}
\end{equation}
From \eqref{2.5} it follows that $U\Pt U^{t} =p_0P^{-1}=\diag(1,p_{0}/p_{1},\dots,p_{0}/p_{d})$ and therefore
\begin{equation}\label{3.10}
\sum_{j=0}^{d}\pt_j u_{i,j}u_{k,j}=\delta_{i,k} \,\frac{p_0}{p_k}.
\end{equation}
Using the last formula we can rewrite the sum over $j$ in \eqref{3.9} as follows
$$\sum_{j=1}^{d}(1-u_{i,j})\pt_{j}u_{k,j}u_{l,j}=-\sum_{j=0}^{d}\pt_{j} u_{i,j} u_{k,j}u_{l,j}.$$
Substituting this into \eqref{3.9}, we conclude that  \eqref{3.4} holds if and only if 
\begin{equation}\label{3.11}
\begin{split}
&\frac{\zeta_{k}-\zeta_{l}}{\al_{k}-\al_{l}}=\sum_{i=1}^{d} \frac{\zeta_{i}-\zeta_{0}}{\al_{i}-\al_{0}}\frac{p_{i}}{p_0}
\sum_{j=0}^{d} \pt_{j} u_{i,j} u_{k,j}u_{l,j},\\
&\qquad\qquad \text{ for all $\zeta=(\zeta_0,\dots,\zeta_d)\in\Rset^{d+1}$ and $k\neq l\in \{0,\dots,d\}$. }
\end{split}
\end{equation}
To complete the proof we need to show that the last equation is equivalent to \eqref{3.5}. From \eqref{3.10} it follows that \eqref{3.11} holds if $k=0$ or $l=0$, so we can assume that $k\neq l\in \{1,\dots,d\} $. For any $i$ different from $0$, $k$ and $l$, the coefficients of $\zeta_i$ on the right-hand side in \eqref{3.11} is equal to zero if and only if \eqref{3.5b} holds. Comparing the coefficients of $\zeta_k$ and $\zeta_l$ on both sides in \eqref{3.11} leads to \eqref{3.5a}. This completes the proof of the implication (a)$\Rightarrow$(b). The opposite direction follows from the above and the fact that the coefficient of $\zeta_0$ on the right-hand side of \eqref{3.11} is $0$ if \eqref{3.5b} holds.
\end{proof}

\begin{Remark}
Equations \eqref{2.8}, \eqref{3.5a}-\eqref{3.5b}, where $\pt_j$ are given in \eqref{2.9}, provide an explicit but rather complicated system of $d$ linear and $3\binom{d}{2}+\binom{d}{3}$ nonlinear equations for the $d^2$ unknown entries $u_{i,j}$, $1\leq i,j\leq d$ of the matrix $U$. When $d=2$, one can show by brute force that the $3$ equations \eqref{2.8} together with \eqref{3.5a} for $k=1$ and $l=2$ imply equation \eqref{3.5a} for $k=2$, $l=1$, and these $4$ equations provide two solutions for the entries $u_{i,j}$. However, proving that this system is consistent for $d= 3$ is a significantly more difficult task even with the use of a computer algebra system. We can bypass this difficulty and construct a solution of these equations by imposing an appropriate ansatz. At this point, we can make a parallel with the work of Gaudin and the Bethe ansatz \cite{Gau2}. Bethe's great insight was to look for eigenstates of the one-dimensional antiferromagnetic Heisenberg model of a specific form, which  depend on appropriate parameters satisfying certain constraints \cite{Bethe}. Gaudin \cite{Gau1} explored this idea for quantum spin chains by constructing Bethe vectors by applying  elementary operations to the vacuum ${\lvert}0{\rangle}$. The operations depend on free parameters, and the Bethe vectors become eigenvectors when the free parameters satisfy the so called Bethe ansatz equations. 
More precisely, using the notations in \cite[page 280, formula (13.20)]{Gau2}, Gaudin considers vectors of the form $S^{-}(\ga_1)S^{-}(\ga_2)\cdots S^{-}(\ga_k){\lvert}0{\rangle}$, where $\ga_1,\dots \ga_k$ are free parameters. The Bethe ansatz equations for the Hamiltonian in $d+1$ spin variables can be written as follows 
\begin{equation*}
\sum_{j=1}^{d+1}\frac{s_j}{\al_j-\ga_i}-\sum_{\begin{subarray}{c}j=1\\ j\neq i \end{subarray}}^{k}\frac{1}{\ga_j-\ga_i}=0, \qquad\text{ for }i=1,\dots, k,
\end{equation*}
see \cite[formula (5.9)]{Gau1} or \cite[page 281, formula (13.27)]{Gau2}.
In particular, if $k=1$, the second sum above is missing and the values of  $\ga_1$ for which $S^{-}(\ga_1){\lvert}0{\rangle}$ is an eigenvector are precisely the roots of the polynomial of degree $d$ given by
\begin{equation}\label{3.12}
\sum_{j=1}^{d+1}\ s_j \prod_{\begin{subarray}{c}\ell=1\\ \ell \neq j  \end{subarray}}^{d+1} (\al_{\ell}-\ga_1)=0.
\end{equation}

We apply a similar idea here by looking for special solutions of equations \eqref{2.8} and \eqref{3.5} of the form 
$$u_{i,j}=\frac{1}{1+\al_i\be_j},$$
where $\be_1,\dots,\be_d$ are free parameters. Equivalently, plugging the above formula for $u_{i,j}$ into \eqref{2.6}, we can define ``Bethe polynomials"  for the multinomial distribution.
With the ansatz above, the huge and complicated system of nonlinear equations for $u_{i,j}$ reduces to a relatively simple set of $d$ algebraic relations  for the unknown parameters $\be_j$ which allow us to identify them with the roots of an explicit polynomial of degree $d$ similar to \eqref{3.12}. What is particularly interesting here is that: 
\begin{enumerate}
\item the roots of this polynomial lead to the construction of multivariate Krawtchouk polynomials which provide a complete eigenbasis of the space $\Rset_N[x]$ of polynomials of total degree at most $N$ for the Gaudin operators, and
\item if we replace the parameters $\al_j$ by the roots $\be_j$ in \eqref{2.14}, we obtain another set of Gaudin operators which are also diagonalized by the same multivariate Krawtchouk polynomials, considered as functions of their degree indices.
\end{enumerate}
This is the content of the main result stated in \thref{th4.1} in the next section.
\end{Remark}

\section{Diagonalization of Gaudin operators}\label{se4}

Since the Gaudin operators $G(\al,p,N;\zeta)$ depend only on the differences $\al_i-\al_0$, we will simplify the formulas by setting 
$\al_0=0$, and thus we will assume below that $\al_{1},\dots,\al_{d}$ are distinct nonzero numbers. 

The main result of the paper is the following theorem.

\begin{Theorem}\label{th4.1}
Let $\al=(\al_0,\dots,\al_d)\in\Rset^{d+1}$, where $\al_j$ are distinct numbers and $\al_0=0$. The polynomial 
\begin{equation}\label{4.1}
R(z)=R(z;p,\al)=p_{0}\prod_{k=1}^{d}(1+\al_{k}z)+\sum_{j=1}^{d} p_{j} \prod_{\begin{subarray}{c}k=1\\ k\neq j \end{subarray}}^{d}(1+\al_{k} z)
\end{equation}
has $d$ distinct nonzero real roots $\be_{1},\dots,\be_{d}$ such that $\al_{i}\be_{j}\neq -1$ for all $i,j\in\{1,\dots,d\}$. We set $\be_{0}=0$, and we define a $(d+1)\times(d+1)$ matrix $U$ with entries 
\begin{equation}\label{4.2}
u_{i,j}=\frac{1}{1+\al_i\be_j}, \qquad 0\leq i,j\leq d,
\end{equation}
and a diagonal matrix $\Pt$ with entries given in \eqref{2.9}. The polynomials $\{P_{n}(x;\kappa,N)\}$ in \eqref{2.6} corresponding to the point $\kappa=(1/p_{0},P,\Pt,U)\in\fK_{d}$ diagonalize the abelian algebras
\begin{itemize}
\item $\mathfrak{G}_{d+1}(\al,p,N)$ (acting on the variables $x_{i}$), and 
\item $\tilde{\mathfrak{G}}_{d+1}(\be;\pt,N)$ (acting on the degree indices $n_{j}$). 
\end{itemize}
Moreover, the following spectral equations hold 
\begin{subequations}\label{4.3}
\begin{align}
G(\al,p,N;\zeta)P_{n}(x;\kappa,N)&=\mu_{n}(\al,\be,p;\zeta)P_{n}(x;\kappa,N), \label{4.3a} \\ 
\Gt(\be,\pt,N;\zeta)P_{n}(x;\kappa,N)&=\mu_{x}(\be,\al,\pt;\zeta)P_{n}(x;\kappa,N),  \label{4.3b}
\end{align}
for all $\zeta\in\Rset^{d+1}$, $x,n\in V_N^d$, where 
\begin{equation}\label{4.3c}
\mu_{n}(\al,\be,p;\zeta) =-\sum_{i=1}^{d}\sum_{j=1}^{d} \frac{(\zeta_{i}-\zeta_{0})p_{i}n_{j}\be_{j }}{1+\al_{i}\be_{j}}.
\end{equation}
\end{subequations}
\end{Theorem}

\begin{proof}
Since $R(0;p,\al)=1$ and $R(-1/\al_{i};p,\al)=p_{i}  \prod_{\begin{subarray}{c}k=1\\ k\neq i \end{subarray}}^{d}(1-\al_{k}/\al_{i}) \neq 0$,  the roots  $\be_{1},\dots,\be_{d}$ of $R(z;p,\al)$ are nonzero and satisfy $\al_{i}\be_{j}\neq -1$. We prove that they are real and distinct by showing that the reverse polynomial 
$$r(z)=z^{d}R(1/z;p,\al)=p_{0}\prod_{k=1}^{d}(z+\al_{k})+ z \sum_{j=1}^{d} p_{j} \prod_{\begin{subarray}{c}k=1\\ k\neq j \end{subarray}}^{d}(z+\al_{k})$$
has $d$ distinct real roots. Since $r(-\al_{k})\neq 0$, the roots of $r(z)$ coincide with the zeros of the rational function $f(z)=r(z)/q(z)$, where $q(z)=\prod_{k=0}^{d}(z+\al_{k})=z\prod_{k=1}^{d}(z+\al_{k})$. Suppose that $(i_{0},i_{1},\dots,i_{d})$ is a permutation of $(0,1,\dots, d)$ which arranges the numbers $-\al_{j}$ in increasing order, i.e.
$-\al_{i_{0}}<-\al_{i_{1}}<\cdots<-\al_{i_{d}}$. Since
$$f(z)=\frac{r(z)}{q(z)}=\frac{p_{0}}{z}+\sum_{k=1}^{d}\frac{p_{k}}{z+\al_{k}}$$
and $f$ has opposite one-sided limits $\lim_{z\to -\al_{i_{k}}^{+}}f(z)=\infty$, $\lim_{z\to -\al_{i_{k+1}}^{-}}f(z)=-\infty$ at the end points of each interval $I_{k}=(-\al_{i_{k}},-\al_{i_{k+1}})$, it follows that $f$ has at least one real zero on each interval $I_{k}$ for $k=0,1,\dots,d-1$, proving that $f$ and $r$ have $d$ distinct real roots.

Clearly, if $\al_0=\be_0=0$, equation \eqref{4.2} implies that $u_{i,0}=u_{0,j}=1$, i.e. $U$ has the form in \eqref{2.4}. We show next that with the ansatz \eqref{4.2}, the conditions in \prref{pr3.2}(b) are satisfied if equation \eqref{2.8} holds. If we set 
\begin{equation}\label{4.4}
w_{k,l}=\frac{\al_k}{\al_k-\al_l},\qquad \text{ for }k\neq l\in \{0,\dots,d\},
\end{equation}
then for every $k\neq l\in \{0,\dots,d\}$ and $j\in\{0,\dots,d\}$ we have 
\begin{equation}\label{4.5}
u_{k,j}u_{l,j}=w_{k,l} u_{k,j}+w_{l,k} u_{l,j}.
\end{equation}
Since $w_{k,l}$ is independent of $j$, it follows that for every $i\neq k,l$ we have 
\begin{align*}
\sum_{j=0}^{d}\pt_{j}u_{i,j}u_{k,j}u_{l,j}=w_{k,l} \sum_{j=0}^{d}\pt_{j}u_{i,j}u_{k,j}+w_{l,k} \sum_{j=0}^{d}\pt_{j}u_{i,j}u_{l,j},
\end{align*}
and therefore the right-hand side is $0$ by \eqref{3.10}. This means that with the ansatz in \eqref{4.2}, equation \eqref{3.5b} is automatically satisfied. Multiplying \eqref{4.5} by $\pt_{j}u_{k,j}$, summing over $j$, and using \eqref{3.10} again we see that
\begin{align*}
\frac{p_{k}}{p_{0}}\sum_{j=0}^{d}\pt_{j}u_{k,j}^{2}u_{l,j}= \frac{p_{k}w_{k,l}}{p_{0}}\sum_{j=0}^{d}\pt_{j}u_{k,j}^{2} + \frac{p_{k}w_{l,k}}{p_{0}}\sum_{j=0}^{d}\pt_{j}u_{k,j}u_{l,j}=w_{k,l}=\frac{\al_k}{\al_k-\al_l},
\end{align*}
i.e. equation \eqref{3.5a} is also automatically satisfied. To complete the proof of \eqref{4.3a} we need to show that if the parameters $\be_1,\dots,\be_{d}$ are the roots of the polynomial $R(z;p,\al)$ then equations \eqref{2.8} hold. But note that $u_{i,j}$ is invariant if we exchange the roles of $i$ and $j$ and the vectors $\al$ and $\be$. Therefore, if we set 
\begin{equation*}
\wt_{k,l}=\frac{\be_k}{\be_k-\be_l},\qquad \text{ for }k\neq l\in \{0,\dots,d\},
\end{equation*}
equation \eqref{4.5} gets replaced by 
\begin{equation*}
u_{j,k}u_{j,l}=\wt_{k,l} u_{j,k}+\wt_{l,k} u_{j,l}.
\end{equation*}
The last equation and arguments similar to the ones above show that equations \eqref{2.8} hold if and only if they hold for $i=0$ and $k\in\{1,\dots,d\}$, leading to the following system of $d$ equations for the parameters $\be_1,\dots,\be_d$
\begin{equation*}
p_{0}+\sum_{j=1}^{d}\frac{p_j}{1+\al_j\be_k}=0 \qquad \text{ for all } k\in\{1,\dots,d\}.
\end{equation*}
It is easy to see that these equations coincide with the equation for the roots of the polynomial $R(z;p,\al)$, completing the proof of \eqref{4.3a}.

Equation~\eqref{4.3b} follows from \eqref{4.3a} and the duality in \eqref{2.12}. Indeed, since transposing the matrix $U$ amounts to exchanging the vectors $\al$ and $\be$ and $U\Pt U^t$ is diagonal, the arguments above show that the polynomials corresponding to the point $\tilde{\kappa}\in \fK_{d}$ diagonalize the operators in the abelian algebra $\mathfrak{G}_{d+1}(\be,\pt,N)$ which combined with equations \eqref{2.12} gives \eqref{4.3b}.
\end{proof}

\begin{Remark}
The proof of \thref{th4.1} shows that the subset of $\fK_{d}$ corresponding to polynomials which diagonalize the Gaudin operators is invariant under the bispectral involution \eqref{2.12a}. It is perhaps interesting to note also that the polynomial in \eqref{4.1} constructed with the dual parameters $\pt_j,\be_j$ must vanish at $\al_1,\dots,\al_d$, hence $R(z;\pt,\be)=\pt_{0}\prod_{k=1}^{d}(1+\be_{k}z)+\sum_{j=1}^{d} \pt_{j} \prod_{\begin{subarray}{c}k=1\\ k\neq j \end{subarray}}^{d}(1+\be_{k} z)=\prod_{k=1}^{d}(1-z/\al_{k})$.
\end{Remark}

\section*{Acknowledgments}
I would like to thank Emil Horozov and Milen Yakimov for useful discussions, and a referee for suggestions to improve an earlier version of the paper.

\end{document}